\documentclass[12pt,english]{article}
\usepackage[T1]{fontenc}
\usepackage[latin9]{inputenc}
\usepackage{geometry}
\usepackage{array}
\usepackage{units}
\usepackage{multirow}
\usepackage{amsthm}
\usepackage{amsmath}
\usepackage{amssymb}
\usepackage{graphicx}
\usepackage{setspace}
\usepackage[authoryear]{natbib}
\PassOptionsToPackage{normalem}{ulem}
\usepackage{ulem}
\makeatletter

\providecommand{\tabularnewline}{\\}
\newcommand{\lyxdot}{.}

 \theoremstyle{definition}
  \newtheorem{example}{\protect\examplename}
\theoremstyle{plain}
\newtheorem{thm}{\protect\theoremname}
  \theoremstyle{plain}
  \newtheorem{cor}{\protect\corollaryname}
  \theoremstyle{plain}
  \newtheorem{lem}{\protect\lemmaname}

\DeclareMathOperator{\N}{N}

\DeclareMathOperator{\pfdr}{pFDR}

\DeclareMathOperator{\lfdr}{LFDR}

\DeclareMathOperator{\reg}{reg}

\DeclareMathOperator{\loo}{L0O}
\DeclareMathOperator{\luo}{L1O}
\DeclareMathOperator{\mdl}{MDL}

\DeclareMathOperator{\lmo}{L\nicefrac{1}{2}O}

\DeclareMathOperator{\talt}{\theta_{alt}}

\DeclareMathOperator{\htloo}{\hat{\theta}^{L0O}}

\DeclareMathOperator{\hitmdl}{\hat{\theta}_{i}^{MDL}}

\DeclareMathOperator{\hitluo}{\hat{\theta}_{i}^{L1O}}

\makeatother

\usepackage{babel}
  \providecommand{\examplename}{Example}
  \providecommand{\lemmaname}{Lemma}
\providecommand{\corollaryname}{Corollary}
\providecommand{\theoremname}{Theorem}

\begin{document}

\title{Empirical Bayes methods corrected for small numbers of tests}

\author{Marta Padilla and David R. Bickel %
\thanks{The authors thank Ye Yang and Zhengmin Zhang for relevant discussions,
Zhenyu Yang for proofreading, and Corey Yanofsky for both. We also
thank the staff at Editage for copy editing the manuscript. This work
was partially supported by the Faculty of Medicine of the University
of Ottawa, by the Canada Foundation for Innovation, and by the Ministry
of Research and Innovation of Ontario. %
}}

\maketitle
~\\
Ottawa Institute of Systems Biology\\
Department of Biochemistry, Microbiology, and Immunology\\
University of Ottawa\\
451 Smyth Rd.\\
Ottawa, Ontario K1H 8M5\\
dbickel@uottawa.ca\\
~\\
\begin{abstract}
Histogram-based empirical Bayes methods developed for analyzing data
for large numbers of genes, SNPs, or other biological features tend
to have large biases when applied to data with a smaller number of
features such as genes with expression measured conventionally, proteins,
and metabolites. To analyze such small-scale and medium-scale data
in an empirical Bayes framework, we introduce corrections of maximum
likelihood estimators (MLE) of the local false discovery rate (LFDR).
In this context, the MLE estimates the LFDR, which is a posterior
probability of null hypothesis truth, by estimating the prior distribution.
The corrections lie in excluding each feature when estimating one
or more parameters on which the prior depends. An application of the
new estimators and previous estimators to protein abundance data illustrates
how different estimators lead to very different conclusions about
which proteins are affected by cancer. 

The estimators are compared using simulated data of two different
numbers of features, two different detectability levels, and all possible
numbers of affected features. The simulations show that some of the
corrected MLEs substantially reduce a negative bias of the MLE. (The
best-performing corrected MLE was derived from the minimum description
length principle.) However, even the corrected MLEs have strong negative
biases when the proportion of features that are unaffected is greater
than 90\%. Therefore, since the number of affected features is unknown
in the case of real data, we recommend an optimally weighted combination
of the best of the corrected MLEs with a conservative estimator that
has weaker parametric assumptions. 
\end{abstract}
\newpage{}

\section{Introduction\label{sec:Introduction}}

\subsection{False discovery rates for genomics applications}

In genomics, new technologies facilitate the simultaneous measurement
of a wide variety of features, up to hundreds of thousands in number.
Examples of such biological features include genes, locations in the
brain, and single-nucleotide polymorphisms (SNPs) in genome-wide association
studies. A multiple testing problem arises in the analysis of data
involving $N$ features $\left\langle X_{1},X_{2},\ldots,X_{N}\right\rangle $
of every individual belonging to two different groups, labeled \emph{treatment}
and \emph{control} for convenience. For the $i$th feature and a corresponding
effect size $\theta_{i}$, a function \emph{T} defines the statistic
$T_{i}=T\left(X_{i}\right)$ that is used to test the null hypothesis
that $\theta_{i}=\theta_{0}$, where $\theta_{0}$ is the parameter
value corresponding to no effect. For example, a common objective
in genomics is to discover the genes that are differentially expressed
between the treatment and control groups of individuals. Thus, gene
expression data analysis involves testing $N$ null hypotheses of
equivalent expression.

Let $A_{i}$ denote the variable indicating whether the $i$th alternative
hypothesis is true. In the case of a two-sided alternative, $A_{i}=1$
if $\theta_{i}\ne\theta_{0}$ but $A_{i}=0$ if $\theta_{i}=\theta_{0}$.
For example, $A_{i}=1$ means the $i$th feature is \uline{a}ffected
by (or \uline{a}ssociated with) the treatment, disease, or other
perturbation. The $i$th null hypothesis corresponds to a\emph{ discovery}
of an effect if the statistic $T_{i}$ falls within some \emph{rejection
region} $\mathcal{T}$, in which case, the $i$th null hypothesis
is rejected. A discovery\emph{ }of an effect is a \emph{false discovery}
if there is no effect $\left(A_{i}=0\right)$; otherwise, it is a
\emph{true discovery} $\left(A_{i}=1\right)$.

The terminology follows \citet{RefWorks:288}, who introduced the
\emph{false discovery rate} (FDR) as an error measure for multiple
testing. Many variants of the FDR can be found in literature, including
the\emph{ Bayesian FDR} \citep{RefWorks:54} or \emph{nonlocal FDR}
(NFDR) \citep{BFDR} and the \emph{local FDR} (LFDR) \citep{RefWorks:53}.
In particular, the NFDR is the probability that a null hypothesis
is true, conditional on its rejection:
\[
\Psi\left(\mathcal{T}\right)=\Pr\left(A_{i}=0\vert T_{i}\in\mathcal{T}\right)=\frac{E\left(N_{0}\left(\mathcal{T}\right)\right)}{E\left(N_{+}\left(\mathcal{T}\right)\right)},
\]
where $N_{0}\left(\mathcal{T}\right)$ denotes the number of false
discoveries and $N_{+}\left(\mathcal{T}\right)$ denotes the total
number of discoveries \citep{efron_large-scale_2010}. ($\Psi$ is
used to abbreviate $\psi\varepsilon\upsilon\delta\acute{\eta}\varsigma$,
pseudo/false). The LFDR for the $i$th feature is defined as the probability
that the null hypothesis is true given the statistic $t_{i}$, the
observed realization of $T_{i}=T(X_{i})$ \citep{efron_large-scale_2010}.
That is,

\begin{equation}
\psi_{i}=\Psi\left(\left\{ t_{i}\right\} \right)=\Pr\left(A_{i}=0\vert T_{i}=t_{i}\right),\label{eq:LFDR-1-1-2}
\end{equation}
which assumes that $T_{i}$ has a common probability density function
$g_{\theta_{0}}$ conditional on the null hypothesis that $\theta_{i}=\theta_{0}$
and another probability density function $g_{\text{\text{{alt}}}}$
conditional on the alternative hypothesis that $\theta_{i}\ne\theta_{0}$.
According to Bayes's theorem,
\begin{equation}
\psi_{i}=P\left(\theta_{i}=\theta_{0}\vert t_{i}\right)=\frac{\pi_{0}g_{\theta_{0}}\left(t_{i}\right)}{g\left(t_{i}\right)},\label{eq:LFDR}
\end{equation}
where $\pi_{0}=P\left(\theta_{i}=\theta_{0}\right)$ is the expectation
value of the proportion of null hypotheses that are true and $g\left(t_{i}\right)$
is the marginal probability density of the test statistic: 
\begin{equation}
g\left(t_{i}\right)=\pi_{0}g_{\theta_{0}}\left(t_{i}\right)+\left(1-\pi_{0}\right)g_{\text{\text{{alt}}}}\left(t_{i}\right).\label{eq:mixture-density}
\end{equation}
As $\pi_{0}$ and $g\left(t_{i}\right)$ are unknown, they are estimated
with empirical Bayesian methods to obtain the estimated LFDR by making
substitutions into equations \eqref{eq:LFDR}-\eqref{eq:mixture-density}.

\subsection{Motivation and overview}

While high-dimensional biology involves measurements over numerous
features, sometimes millions in number, small-dimensional biology
involves measurements over fewer features. Smaller-scale inference
problems arise not only when the total data set represents a small
number of genes, proteins, metabolites, voxels, or other features
\citep[e.g.,][]{10.1371/journal.pone.0009834}, but also when there
are subsets of a large number of features that have something in common
that distinguishes them from the other features in the data set. For
example, \citet[§7]{EFRON2008197} estimated the LFDR for each voxel
as a member of a reference class of 82 voxels at the same physical
location. The measurements of the other 15,461 voxels are less relevant
to the truth of a null hypothesis corresponding to a voxel in the
smaller reference class.

Unfortunately, the statistical methods that have been successfully
applied to large-scale inference problems are not always directly
applicable to inference problems involving considerably smaller dimensions.
In particular, in the estimation of LFDR, commonly used methods of
estimating the unknown parameters $\pi_{0}$ and $g\left(t_{i}\right)$
in equations \eqref{eq:LFDR} and \eqref{eq:mixture-density} involve
the histogram-based estimation of $g_{\text{\text{\text{{alt}}}}}\left(t_{i}\right)$
\citep[e.g., ][]{RefWorks:55,RefWorks:208}. While this is highly
reliable for  data   sets  with several thousand features \citep{mse2008,shrink2009},
it has a high bias for data sets with small numbers of features. Therefore,
special statistical methods are required when the number of features
is too large for conventional hypothesis testing and yet too small
for methods developed for an extremely large number of features. Hence,
we propose new methods for the estimation of the LFDR in small-scale
inference problems.

This paper is organized as follows. First, Section \ref{sec:Data-reduction}
recalls methods of eliminating a nuisance parameter by reducing the
data vector $x_{i}$ of the \emph{i}th feature to a statistic $T\left(x_{i}\right)$
of smaller dimension. Section \ref{sec:Methods} reviews certain known
LFDR estimators and presents the proposed LFDR estimation techniques.
The application of the new LFDR estimators to a data set with 20 proteins
is described in Section 4. The new LFDR estimators are then tested
and compared using simulated data sets, as described in Section 5.
Finally, Section \ref{sec:Discussion} concludes the paper with a
discussion. Asymptotic results are provided in Appendices A and B
to explain the information-theoretic background behind one of the
new estimators and to relate it to maximum likelihood estimation,
respectively.

\section{\label{sec:Data-reduction}Data reduction and likelihood}

Let $x\in\mathcal{X}$ be a vector of measurements of one feature.
Note that since only one feature is considered in this section, the
subscript \emph{{}``i''} is not used, except in Example \ref{exa:absT-stat},
where a generalization to \emph{N} features is shown. The observed
data vector $x\in\mathcal{X}$ is considered a realization of the
random variable $X$ of probability distribution $P_{\theta,\lambda}$
that admits a probability density function $f_{\theta,\lambda}$ with
respect to some dominating measure, where $\theta\in\Theta$ is the
parameter of interest and $\lambda\in\Lambda$ is the nuisance parameter.
In the case of discrete $X$, the density function is defined with
respect to the counting measure on $\mathcal{X}$. For some known
$\theta_{0}\in\Theta$, we have $\theta=\theta_{0}$ under the null
hypothesis or narrow model and $\theta\ne\theta_{0}$ under the alternative
hypothesis or wide model. 

The following two types of likelihood correspond to different ways
of reducing a vector $x$ to a scalar statistic and of eliminating
the nuisance parameter. Which of the two methods is appropriate depends
on the original parametric family $\left\{ f_{\theta,\lambda}:\theta\in\Theta,\lambda\in\Lambda\right\} $
and on which parameter is of interest.

\subsection{Conditional likelihood}

Consider the functions $S$ and $T$ such that $S\left(X\right)$
and $T\left(X\right)$ are statistics that together contain all the
information in $X$. If $S\left(X\right)$ does not depend on $\theta$
and if the probability density function $g_{\theta}=f_{\theta}\left(\bullet\vert S\left(X\right)=S\left(x\right)\right)$
of the data conditional on $S\left(x\right)$, the realized value
of that statistic, does not depend on $\lambda$, then the function
$\ell$ defined by
\begin{equation}
\ell\left(\theta\right)=g_{\theta}\left(T\left(x\right)\right)=f_{\theta}\left(T\left(x\right)\vert S\left(X\right)=S\left(x\right)\right)\label{eq:cond-lik}
\end{equation}
is called the \emph{conditional likelihood function} given $S\left(x\right)$.
In analogy with equation \eqref{eq:marginal}, \citet[§8.2.1]{Severini2000}
has
\[
f_{\theta,\lambda}\left(x\right)=f_{\theta,\lambda}\left(S\left(x\right),T\left(x\right)\right)=g_{\theta}\left(T\left(x\right)\right)f_{\theta,\lambda}\left(S\left(x\right)\right),
\]
where $f_{\theta,\lambda}$ can denote the probably density function
of $X$, $\left\langle S\left(X\right),T\left(X\right)\right\rangle $,
or $S\left(X\right)$, depending on the context.

\begin{example}
\label{exa:severini}\citep[Example 8.47]{Severini2000}. Suppose
that $X_{1}$ is binomial $\left\langle n_{1},\pi_{1}\right\rangle $,
$X_{2}$ is binomial $\left\langle n_{2},\pi_{2}\right\rangle $,
and $X_{1}$ is independent of $X_{2}$. The parameter of interest
is
\[
\theta=\log\frac{\pi_{1}}{1-\pi_{1}}-\lambda,
\]
where $\lambda$ is the nuisance parameter
\[
\lambda=\log\frac{\pi_{2}}{1-\pi_{2}}.
\]
Then,
\[
\log L\left(\theta,\lambda\right)=x_{1}\theta+S\left(x_{1},x_{2}\right)\lambda-n_{1}\log\left(1+e^{\theta+\lambda}\right)-n_{2}\log\left(1+e^{\lambda}\right),
\]
where $S\left(x_{1},x_{2}\right)=x_{1}+x_{2}=s$ is sufficient. Then,
taking $T\left(x_{1},x_{2}\right)=x_{1},$ the conditional log-likelihood
function\emph{ }given $S\left(x_{1},x_{2}\right)$ is

\[
\log\ell\left(\theta\right)=\log g_{\theta}\left(x_{1}\right)=\theta x_{1}-\log K\left(\theta\right),
\]
where 
\[
K\left(\theta\right)=\sum_{j=\max\left(0,s-n_{2}\right)}^{\min\left(n_{1},s\right)}\binom{n_{1}}{j}\binom{n_{2}}{s-j}e^{j\theta}.
\]

\end{example}
Conditional likelihoods are generally available whenever the parameter
of interest is a natural parameter of an exponential family \citep[§10.3]{RefWorks:170}.
For details, see \citet[§8.2.4]{Severini2000}. A recent application
of the conditional likelihood function to genomics data can be found
in \citet{Bickel:enrichment}.

\subsection{Marginal likelihood}

Let $T$ be a measurable function on $\mathcal{X}$. If, for each
$\theta\in\Theta$, the probability density function $g_{\theta}$
of the \emph{statistic }or \emph{reduced data} $T\left(X\right)$
does not depend on the value of $\lambda$, then $\ell\left(\theta\right)=g_{\theta}\left(T\left(x\right)\right)$
defines the \emph{marginal likelihood function }$\ell$.

If, in addition, the conditional distribution of $X$ given $T\left(X\right)=T\left(x\right)$
does not depend on $\theta$, then $T\left(X\right)$ is called \emph{sufficient
}for $\theta$. In that case, no information about $\theta$ is lost
in replacing $X$ with $T\left(X\right)$:
\begin{eqnarray}
f_{\theta,\lambda}\left(x\right) & = & g_{\theta}\left(T\left(x\right)\right)f_{\theta,\lambda}\left(x\vert T\left(X\right)=T\left(x\right)\right)\nonumber \\
 & = & g_{\theta}\left(T\left(x\right)\right)f_{\lambda}\left(x\vert T\left(X\right)=T\left(x\right)\right)\label{eq:marginal}\\
 & = & Cg_{\theta}\left(T\left(x\right)\right),\nonumber 
\end{eqnarray}
where $C$ is constant in $\theta$. The constant is unimportant because
it drops out of likelihood ratios:

\[
\frac{f_{\theta_{1},\lambda}\left(x\right)}{f_{\theta_{0},\lambda}\left(x\right)}=\frac{Cg_{\theta_{1}}\left(T\left(x\right)\right)}{Cg_{\theta_{0}}\left(T\left(x\right)\right)}=\frac{\ell\left(\theta_{1}\right)}{\ell\left(\theta_{0}\right)}
\]
for any value of $\lambda\in\Lambda$. 

\begin{example}
\label{exa:T-stat}Suppose $x$ and $y$ are vectors of $m$ and $n$
values that realize the random variables $X$ and $Y$ of independent
components drawn from normal distributions of unknown means $\xi$
and $\eta$, respectively, and of a common unknown standard deviation
$\sigma$. The parameter of interest is the inverse coefficient of
variation defined by $\theta=\left(\xi-\eta\right)/\sigma$ with $\theta=0$
as the null hypothesis and $\theta\ne0$ as the alternative hypothesis;
the parameter space here is $\Theta=\mathbb{R}^{1}$. A suitable statistic
for data reduction is the two-sample $t$ statistic 
\begin{equation}
T\left(x,y\right)=\frac{\hat{\xi}\left(x\right)-\hat{\eta}\left(y\right)}{\hat{\sigma}\left(x,y\right)\sqrt{m^{-1}+n^{-1}}},\label{eq:t-stat}
\end{equation}
where $\hat{\xi}$, $\hat{\eta}$, and $\hat{\sigma}^{2}$ are the
usual unbiased estimators. Then $g_{\theta}\left(T\left(x,y\right)\right)$,
the probability density of $T\left(X,Y\right)$ evaluated at the observation
$\left\langle x,y\right\rangle $, is the noncentral Student $t$
probability density with $m+n-2$ degrees of freedom and noncentrality
parameter $\left(m^{-1}+n^{-1}\right)^{-1/2}\theta$. \end{example}

The next example encompasses data of multi-dimensional biology.

\begin{example}
\label{exa:absT-stat}Example \ref{exa:T-stat} is extended to $N$
genes, proteins, or other biological features such that $X_{i}\sim\N\left(\xi_{i},\Sigma_{i,m}\right)$
and $Y_{i}\sim\N\left(\eta_{i},\Sigma_{i,n}\right)$ correspond to
the observed outcome $\left\langle x_{i},y_{i}\right\rangle $ for
the $i$th feature, where $i=1,\dots,N$ and $\Sigma_{i,k}$ is the
diagonal covariance matrix of determinant $\sigma_{i}^{2k}$; thus,
$\sigma_{i}$ is the standard deviation of independent measurements
of feature $i$. If whether or not there is an effect on feature $i$
is much more important than the direction of that effect, the parameter
of interest for feature $i$ may be 
\begin{equation}
\theta_{i}=\left|\xi_{i}-\eta_{i}\right|/\sigma_{i},\label{eq:teta-alt-teo}
\end{equation}
the absolute value of the inverse coefficient of variation, with $\theta_{i}=0$
as the null hypothesis, $\theta_{i}>0$ as the alternative hypothesis,
and $\Theta=\left[0,\infty\right)$ as the parameter space. Then $T\left(x_{i},y_{i}\right)$
is the absolute value of the two-sample $t$ statistic for $\left\langle x_{i},y_{i}\right\rangle $
according to equation \eqref{eq:t-stat}, and $T\left(X_{i},Y_{i}\right)$
is distributed as the absolute value of a variate from the noncentral
Student $t$ distribution with $m+n-2$ degrees of freedom and noncentrality
parameter $\delta_{i}=\left(m^{-1}+n^{-1}\right)^{-1/2}\theta_{i}$.
Thus, the density $g_{\theta_{i}}\left(T\left(x_{i},y_{i}\right)\right)$
for the $i$th feature is the probability density of $T\left(X_{i},Y_{i}\right)$
evaluated at $\left\langle x_{i},y_{i}\right\rangle $. \citet{NMWL,smallScale}
illustrated different methods of penalized maximum likelihood estimation
of the LFDR under this model.
\end{example}
\citet[§8.3]{Severini2000} and \citet{RefWorks:127} provide additional
examples of the marginal likelihood, also called the \emph{reduced
likelihood} and not to be confused with the likelihood integrated
with respect to a prior distribution.

\section{\label{sec:Methods}Local false discovery rate estimation}

As mentioned in Section \ref{sec:Introduction}, previous estimators
of FDR and LFDR are highly biased for a moderate or small number of
hypotheses. We present several strategies in this section to reduce
that bias.

\subsection{Previous LFDR estimators\label{sub:Published-methods}}

In this subsection, we review the previous LFDR estimators that lay
the foundations on which our new estimators are constructed.

\subsubsection{LFDR estimates based on other false discovery rates }

Recall from Section \ref{sec:Introduction} that the $i$th null hypothesis
is rejected if the statistic $t_{i}$ falls within some \emph{rejection
region} $\mathcal{T}$. To avoid the specification of such a rejection
region $\mathcal{T}$, an estimated q-value $q\left(p_{i}\right)$
is commonly calculated for the $i$th p-value $p_{i}$ among the $N$
p-values. The rejection region $\mathcal{T}_{\alpha}$ is a function
of the significance level $\alpha$, the usual Type I error rate of
rejecting the $i$th null hypothesis if and only if $p_{i}\le\alpha$;
thus, the estimated q-values, herein called \emph{q-values} to follow
contemporary terminology \citep{ISI:000272935000021}, are given by
\begin{equation}
q\left(p_{i}\right)=\min_{\alpha\in\left[p_{i},1\right]}\widehat{\pfdr}\left(\mathcal{T}_{\alpha}\right);\,\left[i=1,\dots,N\right],\label{eq:q-value}
\end{equation}
where $\widehat{\pfdr}\left(\mathcal{T}_{p_{i}}\right)$ is an estimate
of the positive FDR (pFDR) on the rejection region $\mathcal{T}_{p_{i}}$
\citep{RefWorks:282}. Thus, the q-value is the lowest estimated pFDR
at which the $i$th null hypothesis is rejected. Because the latter
effectively uses 1 as an estimate of $\pi_{0}$, it will be called
QV1 in order to distinguish it from $q\left(p_{i}\right)$, which
is called QV. 

In addition, conservative LFDR estimators based on the binomial distributions
have been proposed by \citet{BFDR}. The estimator that \citet{BFDR}
called the {}``MLE'' is renamed in this paper to avoid confusion
with the estimator addressed in the next subsection. We denote the
version that uses the estimate of $\pi_{0}$ described in \citet{RefWorks:282}
as the \emph{binomial-based estimator} (BBE) to distinguish it from
BBE1, which instead uses 1 as an estimate of $\pi_{0}$.

\subsubsection{\label{sub:Maximum-likelihood-estimator}Maximum likelihood estimator}

The maximum likelihood estimator (MLE) described in this subsection
will be called the \emph{leave-zero-out} (L0O) method for reasons
given in Section \ref{sub:Evaluation-LFDR}. The LFDR is estimated
under the assumption that both the null-hypothesis density function
$g_{\theta_{0}}$ and the alternative-hypothesis density function
$g_{\text{\text{{alt}}}}$ of equations \eqref{eq:LFDR}-\eqref{eq:mixture-density}
are members of $\left\{ g_{\theta}:\theta\in\Theta\right\} $, a parametric
family of probability density functions indexed by the interest parameter
$\theta$, which is a member of some parameter space $\Theta$. Thus,
$g_{\text{alt}}=g_{\theta_{\text{alt}}}$, where $\theta_{\text{alt}}\in\Theta$
is unknown and not equal to the known $\theta_{0}\in\Theta$. Any
nuisance parameter must be eliminated, perhaps by using one of the
two methods explained in Section \ref{sec:Data-reduction}. 

For the $i$th feature, the data vector $x_{i}$ is reduced to a scalar
statistic $t_{i}$, as in Examples \ref{exa:severini}-\ref{exa:absT-stat}.
Therefore, $g_{\theta_{0}}\left(t_{i}\right)$ and $g_{\theta_{\text{{alt}}}}\left(t_{i}\right)$
denote the probability densities for the reduced data under the null
hypothesis and the alternative hypothesis, respectively. The true
value of the LFDR for the $i$th feature is, according to equations
\eqref{eq:LFDR}-\eqref{eq:mixture-density} and $g_{\text{{alt}}}=g_{\theta_{\text{alt}}}$,

\begin{equation}
\psi_{i}=\frac{\pi{}_{0}g_{\theta_{0}}\left(t_{i}\right)}{\pi{}_{0}g_{\theta_{0}}\left(t_{i}\right)+\left(1-\pi{}_{0}\right)g_{\theta_{\text{{alt}}}}\left(t_{i}\right)},\label{eq:LFDR-true}
\end{equation}
which is unknown since $\theta_{\text{{alt}}}$ and $\pi{}_{0}$ are
unknown.

The L0O method involves the estimation of the parameters $\pi_{0}$
and $\theta_{\text{alt}}$. These estimated parameters $\left\langle \hat{\theta}^{\loo},\hat{\pi}_{0}^{\loo}\right\rangle $
are the maximum likelihood estimates of the true parameters given
by 
\begin{equation}
\left\langle \hat{\theta}^{\loo},\hat{\pi}_{0}^{\loo}\right\rangle =\arg\sup_{\left\langle \theta,\pi_{0}\right\rangle \in\Theta\times\left[0,1\right]}\prod_{j=1}^{N}\left(\pi_{0}g_{\theta_{0}}\left(t_{i}\right)+\left(1-\pi_{0}\right)g_{\theta}\left(t_{i}\right)\right).\label{eq:maxlik-L0O}
\end{equation}
Therefore, with substitution into equation \eqref{eq:LFDR}, the estimated
LFDR for the \emph{i}th feature is 

\begin{equation}
\hat{\psi}_{i}^{\loo}=\frac{\hat{\pi}_{0}^{\loo}g_{\theta_{0}}\left(t_{i}\right)}{\hat{\pi}_{0}^{\loo}g_{\theta_{0}}\left(t_{i}\right)+\left(1-\hat{\pi}_{0}^{\loo}\right)g_{\htloo}\left(t_{i}\right)}.\label{eq:LFDR-L0O}
\end{equation}
This estimator has been used with marginal likelihood \citep{GWAselect,smallScale}
and conditional likelihood \citep{Bickel:enrichment}. Similarly,
\citet{ParametricMixtureLFDR} had estimated the LFDR by maximizing
the likelihood over exponential families.

\subsection{New LFDR estimators\label{sub:Proposed-strategies}}

Here, 5 novel LFDR estimators are proposed: 3 are corrected MLEs,
and the other 2 are related to the BBE. The corrected MLEs are based
on equation \eqref{eq:LFDR}. The fourth technique is an approximation
of the BBE, and the last new estimator is a combination of the BBE
and one of the corrected MLEs.

\subsubsection{Corrected MLEs}

The three methods presented here correct the bias of the L0O method
that results from using the same statistic $t_{i}$ to evaluate the
density functions and to estimate $\pi_{0}$ and $\theta_{\text{alt}}$.
This is accomplished by removing dependence of the estimators on $t_{i}$
prior to evaluating the density functions at $t_{i}$. While that
negative bias vanishes as the number of features increases (Appendix
B), it can be unacceptably large for small numbers of features.

The first corrected MLE is called the \emph{minimum description length}
(MDL) method. Although the method was inspired by the MDL principle
(Appendix A), the general idea of estimating a prior on the basis
of exchangeable features other than the feature under consideration
is implicit in \citet{Goodman2004b}; cf. \citet{Gastpar2010890}
and J. Cuzick's discussion of \citet{Aitkin1991b}. The MDL method
uses modified estimates of parameters $\pi_{0}$ and $\theta_{\text{\text{{alt}}}}$
for the $i$th feature, denoted as $\left\langle \hat{\theta}_{i}^{\mdl},\hat{\pi}_{0i}^{\mdl}\right\rangle $
for $i\in\left\{ 1,\dots,N\right\} $. These estimated parameters
are obtained by maximizing the likelihood function:

\begin{equation}
\left\langle \hat{\theta}_{i}^{\mdl},\hat{\pi}_{0i}^{\mdl}\right\rangle =\arg\sup_{\left\langle \theta,\pi_{0}\right\rangle \in\Theta\times\left[0,1\right]}\prod_{j=1,\, j\ne i}^{N}\left(\pi_{0}g_{\theta_{0}}\left(t_{i}\right)+\left(1-\pi_{0}\right)g_{\theta}\left(t_{i}\right)\right).\label{eq:maxlik-MDL}
\end{equation}
Note that the product is obtained over all features except for the
\emph{i}th feature. Accordingly, the MDL estimator of LFDR for the
\emph{i}th feature is given by 
\begin{equation}
\hat{\psi}_{i}^{\mdl}=\frac{\hat{\pi}_{0i}^{\mdl}g_{\theta_{0}}}{\hat{\pi}_{0i}^{\mdl}g_{\theta_{0}}\left(t_{i}\right)+\left(1-\hat{\pi}_{0i}^{\mdl}\right)g_{\hitmdl}\left(t_{i}\right)}.\label{eq:LFDR-MDL}
\end{equation}

The second corrected MLE estimator, called \emph{leave-one-out} (L1O),
is the same as the MDL except that the L0O estimate of $\pi_{0}$
is used instead of $\hat{\pi}_{0i}^{\mdl}$. Therefore, in L1O, three
steps are involved. First, the parameters $\left\langle \hat{\theta}^{\loo},\hat{\pi}_{0}^{\loo}\right\rangle $
are calculated from the likelihood function \eqref{eq:maxlik-L0O}
involved in the L0O method, which includes all the features. Second,
similar to MDL, the likelihood function involving all features, except
for the \emph{i}th feature, is maximized for every feature using the
$\hat{\pi}_{0}^{\loo}$ obtained in the previous step. Therefore,
in this step, the interest parameter for all $i\in\left\{ 1,\dots,N\right\} $
is estimated as 

\begin{equation}
\hat{\theta}_{i}^{\luo}=\arg\sup_{\theta\in\Theta}\prod_{j=1,\, j\ne i}^{N}\left(\hat{\pi}_{0}^{\loo}g_{\theta_{0}}\left(t_{j}\right)+\left(1-\hat{\pi}_{0}^{\loo}\right)g_{\theta}\left(t_{j}\right)\right),\label{eq:maxlik-L1O}
\end{equation}
leading to the L1O estimator of LFDR for the \emph{i}th feature:

\begin{equation}
\hat{\psi}_{i}^{\luo}=\frac{\hat{\pi}_{0}^{\loo}g_{\theta_{0}}}{\hat{\pi}_{0}^{\loo}g_{\theta_{0}}\left(t_{i}\right)+\left(1-\hat{\pi}_{0}^{\loo}\right)g_{\hitluo}\left(t_{i}\right)}.\label{eq:LFDR-L1O}
\end{equation}

The MDL and L1O strategies eliminate bias from a double use of feature
data. However, when there is only a single affected feature, the MDL
and L1O do not use any information about $\theta_{\text{\text{{alt}}}}$
to estimate the LFDR of the only affected feature, introducing considerable
bias in estimating $\theta_{\text{\text{{alt}}}}$. 

To overcome this defect, we introduce the third corrected MLE, called
the \emph{leave-half-out} ($\lmo$) estimator. Like L1O, $\lmo$ includes
information about the \emph{i}th feature through $\hat{\pi}_{0}^{\loo}$;
furthermore, half of the information about each left-out feature is
also included in the $\lmo$ through the likelihood function. Such
a function is a \emph{weighted likelihood function}, where the contribution
of the \emph{i}th feature to the overall likelihood function is corrected
by a \emph{weight} $w_{ij}$ given by
\[
w_{ii}\left(\nu\right)=\frac{\nu}{\nu+(N-1)};\, w_{ij}\left(\nu\right)=\frac{1}{\nu+(N-1)}\,\left[j\ne i\right],
\]
where $\nu\in\left[0,1\right]$ is the information (log-likelihood)
weight of $t_{i}$ relative to each $t_{j}$ for the purpose of estimating
the parameter of interest. Thus, the weights satisfy $\underset{j=1}{\overset{N}{\sum}}w_{ij}=1$.
The $\nu$-weighted likelihood function for feature $i$ is

\begin{equation}
L_{i}(\pi_{0},\,\theta_{\text{{alt}}};\, t_{i},\nu)=\prod_{j=1}^{N}\left(\pi_{0}g_{\theta_{0}}\left(t_{j}\right)+\left(1-\pi_{0}\right)g_{\theta_{\text{{alt}}}}\left(t_{j}\right)\right){}^{w_{ij}\left(\nu\right)},\label{eq:wlik}
\end{equation}
and the maximum $\nu$-weighted likelihood is
\begin{equation}
\hat{\theta}_{i}^{\text{L}\nu\text{O}}=\arg\sup_{\theta\in\Theta}L_{i}(\hat{\pi}_{0}^{\loo},\,\theta;\, t_{i},\nu),\label{eq:mwlik}
\end{equation}
degenerating to the L0O and L1O estimators when $\nu=1$ and $\nu=0$,
respectively. Thus, $\hat{\theta}_{i}^{\text{L}\nu\text{O}}$ may
be considered the \emph{leave-$\nu$-out }($\text{L}\nu\text{O}$)
estimator.

The proposed $\lmo$ method includes exactly half of the information
about the \emph{i}th feature in its likelihood function by setting
$\nu=1/2$. Therefore, the new estimator $\hat{\theta}_{i}^{\lmo}$
for all $i\in\left\{ 1,\dots,N\right\} $ is given by the maximization
of the weighted likelihood function according to equations \eqref{eq:wlik}-\eqref{eq:mwlik}.
With such estimated parameters $\left\langle \hat{\theta}_{i}^{\lmo},\hat{\pi}_{0}^{\loo}\right\rangle $,
the LFDR for the \emph{i}th feature ($\hat{\psi}_{i}^{\lmo}$) can
be estimated with equation \eqref{eq:LFDR-L1O}, replacing $\hat{\theta}_{i}^{\luo}$
with $\hat{\theta}_{i}^{\lmo}$, and analogously for $\hat{\psi}_{i}^{\text{L}\nu\text{O}}$
given any $\nu$ between 0 and 1.

Weighted likelihoods have been reviewed by \citet{RefWorks:501} and
applied to the quantification of evidence by \citet{NMWL}.

\subsubsection{BBE-related LFDR estimators\label{sub:BBE-related}}

A method for approximating the BBE \citep{BFDR} is also presented
here. BBE attempts to estimate the LFDR more conservatively than q-values,
which were not originally designed for LFDR estimation. In this section,
we denote the q-values as $q$, which refers to either QV or QV1 (see
Section \ref{sec:Methods}), and $\rho_{i}$ denotes the rank of the
q-values corresponding to the \emph{i}th feature, such that $q_{(1)}\leq q_{(2)}\leq...\leq q_{(N)}$.
The new proposed method directly assigns twice the rank of the q-value
$q_{(2\rho_{i})}$ to the LFDR estimate of the \emph{i}th feature
with the corresponding q-value $q_{(\rho_{i})}$. Therefore, we define
\emph{(estimated) r-values} as
\begin{equation}
r(q_{i})=\begin{cases}
q_{(2\rho_{i})} & \text{{if}}\:\rho_{i}\leq N/2\\
1 & \text{{if}}\:\rho_{i}>N/2.
\end{cases}
\end{equation}
We employ analogous notation for r-values, i.e., RV when it uses QV
and RV1 when it uses QV1.  Our aim is to verify that RV and RV1
approximate BBE and BBE1, respectively. 

Finally, for reasons given in Section \ref{sub:Evaluation-LFDR},
we combine BBE and MDL into an estimator that leverages the strengths
of each. Specifically, the \emph{MDL-BBE} is the linearly combination
of the other two estimators with weights that are optimal for the
hedging game of \citet{Bickel:combine}.

\section{\label{sec:Application}Application}

In Alex Miron\textquoteright{}s laboratory at the Dana-Farber Cancer
Institute, the abundance levels of 20 plasma proteins of 55 women
with HER2-positive breast cancer, 35 women with ER/PR-positive breast
cancer, and 64 healthy women \citep{ProData2009b} were measured.
The respective data vectors $x_{1}^{\text{HER2}},\dots,$ $x_{20}^{\text{HER2}}$,
$x_{1}^{\text{ER/PR}},\dots,x_{20}^{\text{ER/PR}}$, $y_{1},\dots,y_{20}$
were created by adding the first quartile of the abundance levels
(over the 64 healthy women and over all proteins) to each abundance
level and by taking natural logarithms of the resulting sums; similar
conservative prepossessing steps have worked well with gene expression
data \citep{RefWorks:955}. 

The preprocessed data were modeled as normally distributed, as illustrated
in Example \ref{exa:absT-stat}. Following the notation of the example,
$\xi_{i}^{\text{HER2}}$, $\xi_{i}^{\text{ER/PR}}$, and $\eta_{i}$
are the expectation values of $X_{i}^{\text{HER2}}$, $X_{i}^{\text{ER/PR}}$,
and $Y_{i}$, respectively, and are as such interpretable as population
levels of the abundance of protein $i$. The parameters of interest
are $\theta_{i}^{\text{HER2}}=\left|\xi_{i}^{\text{HER2}}-\eta_{i}\right|/\sigma_{i}$
and $\theta_{i}^{\text{ER/PR}}=\left|\xi_{i}^{\text{ER/PR}}-\eta_{i}\right|/\sigma_{i}$,
the standardized levels of the $i$th protein's abundance relative
to the healthy controls. In this context, the LFDR of each protein
is a posterior probability that its average abundance level is not
affected by cancer.

The data were analyzed according to the distributions of $T\left(X_{i}^{\text{HER2}},Y_{i}\right)$
and $T\left(X_{i}^{\text{ER/PR}},Y_{i}\right)$ given in Example \ref{exa:absT-stat}.
The methods of estimating the LFDR described in Section \ref{sub:Published-methods}
were applied to the proteomics data, namely, MDL, L0O, L1O, $\lmo$,
BBE, BBE1, RV, RV1, and MDL-BBE. The results are shown in Figures
\ref{fig:Volcano} and \ref{fig:lfdr-pval}, which represent LFDR
versus the estimated protein abundance ratio and p-value, respectively.
All figures show results for the HER2-positive and ER/PR-positive
groups separately. The volcano plot (Figure \ref{fig:Volcano}) indicates
that the proteins most affected by cancer, showing estimated abundance
ratios furthest from unity, have LFDR estimates close to zero, while
higher LFDR estimates correspond to proteins with estimated abundance
ratios close to unity. From the results shown in both figures, we
can see that the selection of the LFDR estimator is crucial because
for thresholds of the estimated LFDR between 0 and 0.2, many proteins
would be considered affected or unaffected by cancer, depending on
the method. BBE1 and RV1 were omitted from the figures to ensure legibility.

\begin{figure}
\includegraphics[width=17cm]{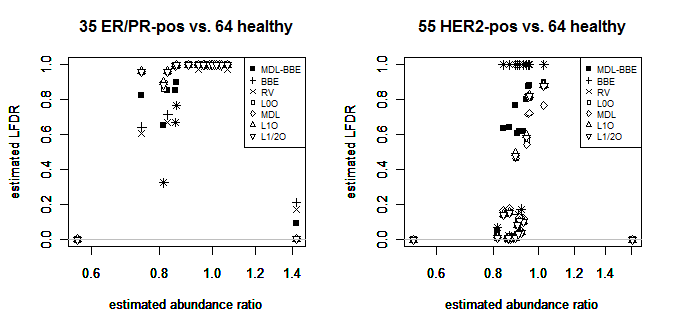}\caption{Volcano plot representing LFDR for protein abundance of both groups,
HER2-positive and ER/PR-positive women, relative to healthy women,
estimated by using different LFDR estimators and represented versus
the estimated protein abundance ratio. The LFDR estimators are MDL,
L0O, L1O, $\lmo$, BBE, RV, and MDL-BBE.\label{fig:Volcano}}
\end{figure}
\begin{figure}
\includegraphics[width=17cm]{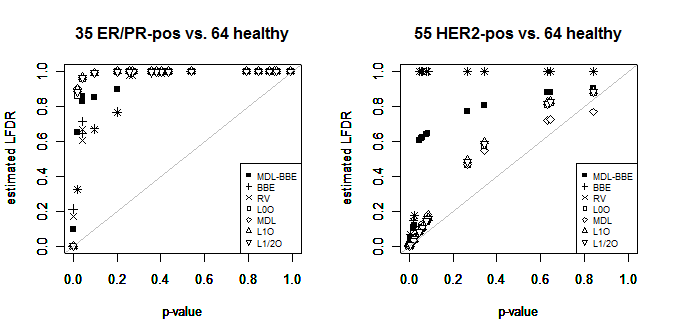}\caption{LFDR for protein abundance of both groups of women with breast cancer,
HER2-positive and ER/PR-positive, relative to healthy women, estimated
by using different estimators and represented versus p-value. The
LFDR estimators are MDL, L0O, L1O, $\lmo$, BBE, RV, and MDL-BBE.
\label{fig:lfdr-pval}}
\end{figure}

\section{\label{sec:Simulations}Simulations}

In this section, the performance of the LFDR estimators described
in Section \ref{sub:Published-methods} is compared using simulated
protein abundance data. Such methods are MDL, L0O, L1O, $\lmo$, BBE,
BBE1, RV, RV1, and a combination of MDL and\emph{ }BBE. The design
of each data set is patterned after that of Sections \ref{sub:Maximum-likelihood-estimator}
and \ref{sec:Application}. It consists of abundance levels of $N$
proteins for two groups,\emph{ sick} and\emph{ healthy,} each containing
5 individuals, for total of $10$ abundance levels per protein. For
the \emph{i}th protein, the log-abundance data are drawn from a normal
distribution with variance $\sigma^{2}=1$ and mean equal to 0, except
for the proteins affected by the disease, which have mean $\xi_{\text{{alt}}}>0$
in the sick group. To represent both barely detectable and highly
detectable differences between the null and alternative distributions,
we consider two values for the effect size, a low value $\left(\xi_{\text{{alt}}}=1.5\right)$
and a high value $\left(\xi_{\text{{alt}}}=4\right)$ relative to
the standard deviation $\left(\sigma=1\right)$. Therefore, we have
two values for the positive noncentrality parameter $\delta_{\text{{alt}}}=\left(m^{-1}+n^{-1}\right)^{-1/2}\theta_{\text{{alt}}}$,
where, in agreement with equation \eqref{eq:teta-alt-teo}, 
\begin{equation}
\theta_{\text{{alt}}}=|\xi_{\text{{alt}}}-0|/\sigma=\xi_{\text{{alt}}},\label{eq:teta-alt}
\end{equation}
and \emph{m} and \emph{n} are the numbers of individuals in the sick
and healthy group, respectively $\left(m=n=5\right)$. Therefore,
the distribution of the affected proteins in the sick group has $\delta_{\text{{alt}}}=2.4$
in the low-effect simulations and $\delta_{\text{{alt}}}=6.3$ in
the high-effect simulations. By contrast, the noncentrality parameter
values are 0 for all the unaffected proteins and for all the proteins
of the healthy group. Then, the LFDR estimators are compared with
regard to the number of proteins in each data set and the number of
affected proteins for 20 simulated data sets of each configuration. 

To facilitate the comparison among the different LFDR estimators and
for specific verification of the similarities between BBE and RV and
BBE1 and RV1, we estimated each estimator's \emph{bias}, the mean
(over all proteins) of the expectation value of the difference between
the estimate and true LFDR. For each LFDR estimator, that bias is
estimated by the mean difference between the estimated LFDR and the
true LFDR, where the mean is over the simulations as well as the proteins.
Thus, 60 LFDR estimates are averaged when the data set has 3 proteins
(mean over 3 proteins and over 20 simulations) and 300 LFDR values
when the data set has 15 proteins (mean over 15 proteins and over
20 simulations). The true value is calculated using equation \eqref{eq:LFDR-true}
with the proportion of proteins that are unaffected as $\pi_{0}$
and with the value of $\theta_{\text{{alt}}}$ given by equation \eqref{eq:teta-alt}. 

The results are shown in Figure \ref{fig:sims}, where the estimated
bias of the LFDR is represented as a function of the number of affected
proteins, for each number of proteins in the data set and for two
different levels of detectability. Although we studied the behavior
of the methods separately for affected and unaffected proteins, the
figures show the estimated biases of the LFDR averaged over all proteins.
Figure \ref{fig:sims}, plots (a) and (b), show the results for a
data set with 3 and 15 proteins, respectively, and for the high level
of detectability. Plots (c) and (d) are the same except that they
correspond to the low detectability level. For better legibility
of the figures, RV1 and BBE1 are not displayed because they have excessively
high estimated bias averaged over either affected or unaffected proteins
or averaged over all proteins. 

It can be seen from Figure \ref{fig:sims} that the LFDR estimates
depend on the number of proteins, the number of affected proteins,
and the detectability level. Note that in the plots, the contribution
of the bias from the affected proteins increases as the number of
affected proteins increases because the protein-averaged results are,
in effect, weighted according to the number of affected or unaffected
proteins. The estimators BBE1 and RV1 are not displayed because the
values of their (positive) biases are much higher than those of the
other estimators. However, the biases of RV and BBE are more moderate,
especially when few proteins are affected.

\section{\label{sec:Discussion}Discussion}

\subsection{Evaluation of the LFDR estimators\label{sub:Evaluation-LFDR}}

Some differences in estimator performance depend on the value of $\delta_{\text{{alt}}}$,
the noncentrality parameter. L0O and L$\nicefrac{1}{2}$O work very
well when $\delta_{\text{{alt}}}$ is high, regardless of the number
of features in the data set (Figure \ref{fig:sims}, (a)-(b)) and
when there is at least one affected feature. When there is no affected
feature, both estimators have highly negative bias (about $-$0.25).
When $\delta_{\text{{alt}}}$ is high, MDL and L1O perform similarly
to L0O and L$\nicefrac{1}{2}$O, except when there is only one affected
feature in the data set, in which case MDL and L1O have excessively
high positive biases for the affected feature. Those biases are not
seen in the plots since they are averaged over all features. These
biases result from the fact that MDL and L1O do not use the data of
the given feature to estimate $\delta_{\text{{alt}}}$. Thus, MDL
and L1O cannot effectively estimate $\delta_{\text{{alt}}}$ when
only one feature is affected, which results in such a noticeable high
positive bias when $\delta_{\text{{alt}}}$ is high. L$\nicefrac{1}{2}$O
overcomes that drawback by including half the information on the unique
affected feature in its likelihood function \eqref{eq:wlik}. In contrast,
BBE and RV are less biased than the other estimators when no features
are affected. However, the values of their conservative (positive)
biases increase with the number of affected features. On the other
hand, when $\delta_{\text{{alt}}}$ is low (Figure \ref{fig:sims},
(c)-(d)), all the corrected MLEs are negatively biased when there
is no affected feature, and BBE and RV have positive biases.

In addition, by comparing the four plots in Figure \ref{fig:sims},
we can see that BBE and RV are extremely similar; only slight differences
appear in cases of few affected proteins. Moreover, the methods gave
similar estimates in the application to real protein data (Figures
\ref{fig:lfdr-pval} and \ref{fig:Volcano}).%
\footnote{However, we found in unpublished work that these estimators diverge
more for an application to a large-scale proteomics data set.%
}

Therefore, since BBE-related estimators show a small bias for no or
a few affected features and since corrected MLEs perform better when
most of the features of the data set are affected, we consider a new
LFDR estimator as the weighted combination of representative estimators
of each type (corrected MLEs and BBE-related estimators): the MDL
and the BBE. Based on performance with 3 features and low $\delta_{\text{{alt}}}$,
we choose to combine MDL and BBE because MDL has the lowest absolute
value of the bias among the corrected MLEs and because the BBE is
simpler than the RV but is similar in performance. Then we applied
the same MDL-BBE method to all cases. The MDL-BBE is an optimal linear
combination of the MDL and the BBE (Section \ref{sub:BBE-related}).

To summarize the findings for each method and each total number of
features in the data set, Table \ref{tab:bias-range} reports the
most extreme values and the median of the biases for $\pi_{0}\geq90$\%
over the numbers of affected features and over both values of $\delta_{\text{{alt}}}$.
We can see from this table that these values are very similar among
the methods of the same type. Corrected MLEs have the most negative
biases, and BBE-related methods have the highest positive biases.
As Table \ref{tab:bias-range} indicates, the MDL-BBE succeeds in
substantially reducing the negativity of the worst-case bias of the
MDL and substantially reducing the highly conservative worst-case
bias of the BBE. In short, the MDL-BBE does not suffer from the main
drawbacks of the other estimators.

Since the focus on the worst-case performance can lead to an overly
pessimistic assessment of small-scale estimation of the LFDR, the
median values are also reported in Table \ref{tab:bias-range}. They
indicate that while estimation is somewhat unreliable for some estimators
when there are only 3 features, it is reliable for all estimators
when there are 15 features. Even so, the reported absolute values
of the biases should be regarded as lower bounds since they were computed
under the independence of features. Further, since the simulations
use the same family of distributions as the MLE-related estimators,
they perform better in the simulations that they would with real data.

\begin{table}
\begin{centering}
\begin{tabular}{|c|c|c|c|c|}
\hline 
\multirow{1}{*}{\textbf{\footnotesize LFDR}} & \multicolumn{2}{c|}{\textbf{\footnotesize all}{\footnotesize{} $\pi_{0}$}} & \multicolumn{2}{c|}{{\footnotesize $\pi_{0}\geq$}\textbf{\footnotesize 90\%}}\tabularnewline
\cline{2-5} 
\textbf{\footnotesize Estimators} & \textbf{\footnotesize 3 features} & \textbf{\footnotesize 15 features} & \textbf{\footnotesize 3 features} & \textbf{\footnotesize 15 features}\tabularnewline
\hline 
\textbf{\footnotesize MDL-BBE} & \textbf{\footnotesize 0.13 {[}$-$0.10, 0.41{]}} & \textbf{\footnotesize 0.01 {[}$-$0.20, 0.31{]}} & \textbf{\footnotesize $-$0.1} & \textbf{\footnotesize $-$0.13 {[}$-$0.2, 0.01{]}}\tabularnewline
\hline 
\textbf{\footnotesize BBE} & \textbf{\footnotesize 0.19 {[}$-$0.07, 0.69{]}} & \textbf{\footnotesize 0.01 {[}$-$0.11, 0.55{]}} & \textbf{\footnotesize $-$0.07} & \textbf{\footnotesize $-$0.07 {[}$-$0.11,$-$0.01{]}}\tabularnewline
{\footnotesize RV} & {\footnotesize 0.18 {[}$-0.08$, 0.69{]}} & {\footnotesize 0.00 {[}$-0.13$, 0.55{]}} & {\footnotesize $-$0.08} & {\footnotesize $-$0.09 {[}$-$0.13,$-$0.02{]}}\tabularnewline
\hline 
\textbf{\footnotesize MDL} & \textbf{\footnotesize 0.02 {[}$-$0.13, 0.12{]}} & \textbf{\footnotesize 0.00 {[}$-$0.30, 0.07{]}} & \textbf{\footnotesize $-$0.13} & \textbf{\footnotesize $-$0.19 {[}$-$0.30, 0.02{]}}\tabularnewline
{\footnotesize L0O} & {\footnotesize $-0.02$ {[}$-0.22$, 0.16{]}} & {\footnotesize $-0.01$ {[}$-0.34$, 0.08{]}} & {\footnotesize $-$0.18} & {\footnotesize $-$0.17 {[}$-$0.26,$-$0.01{]}}\tabularnewline
{\footnotesize L1O} & {\footnotesize 0.02 {[}$-0.17$, 0.20{]}} & {\footnotesize 0.00 {[}$-0.30$, 0.08{]}} & {\footnotesize $-$0.17} & {\footnotesize $-$0.16 {[}$-$0.24, 0.03{]}}\tabularnewline
{\footnotesize L1/2O} & {\footnotesize $-0.02$ {[}$-0.22$, 0.18{]}} & {\footnotesize $0.00$ {[}$-0.32$, 0.08{]}} & {\footnotesize $-$0.17} & {\footnotesize $-$0.17 {[}$-$0.24,$-$0.01{]}}\tabularnewline
\hline 
\end{tabular}
\par\end{centering}

\caption{\label{tab:bias-range}Median {[}minimum, maximum{]} values of the
biases of all the LFDR estimators over all $\pi_{0}$ and over all
$\pi_{0}\geq90$ \%, over the numbers of affected features, and over
both values of the noncentrality parameter. Separate values are given
for each total number of features in the data set.}
\end{table}

\begin{figure}
\includegraphics[width=17cm]{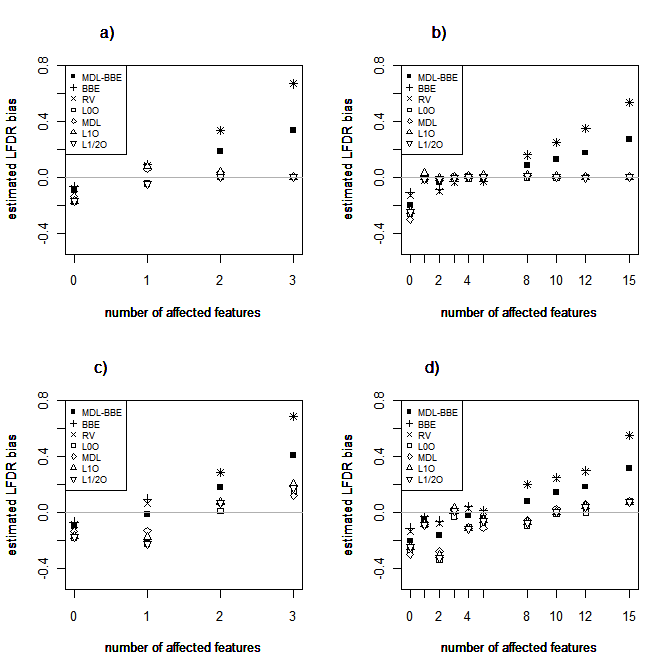}\caption{Estimated bias of several LFDR estimators for an artificial data set
with 3 features ((a) and (c)) and 15 features ((b) and (d)) and cases
of high $\delta_{\text{{alt}}}=6.3$ ((a) and (b)) and low noncentrality
parameter $\delta_{\text{{alt}}}=2.4$ ((c) and (d)) versus increasing
number of affected features. The LFDR estimators are MDL, L0O, L1O,
$\lmo$, BBE, RV, and a combination of MDL and\emph{ }BBE.\label{fig:sims}}
\end{figure}

\subsection{Conclusions}

In this paper, we proposed several LFDR estimators to give reliable
results for small-scale inference. We compared them on simulated data
sets and illustrated their use on a protein-abundance data set that
illustrates that different conclusions would be drawn on the basis
of different estimators. The performance of such methods depends on
the number of features, number of affected features, and values of
the unknown parameters. Simulations showed that the corrected MLEs
have very low biases in all cases when more than 50\% of the features
in the data set are affected, even for a data set with only 3 features.
However, when the proportion of affected features is very small, these
methods have excessively negative biases. In contrast, BBE and RV
have excessively large biases when there is a high proportion of affected
features. Furthermore, this bias increases as the number of affected
features increases in the data set. Therefore, the weighted combination
of an adjusted-MLE (MDL) and a conservative method (RV or BBE)\emph{
}may represent the safest solution for a general scenario in which
the number of affected features is unknown. 
\section*{Colophon}

We used the following packages of R \citep{R2008}: \texttt{Biobase}
\citep{RefWorks:161} and \texttt{qvalue} \citep{Rqvalue} from Bioconductor;
\texttt{locfdr} \citep{RefWorks:208}, \texttt{fBasics} \citep{RfBasics},
and \texttt{distr} \citep{Rdistr} from the CRAN repository.

\section*{\label{sec:Appendix-A}Appendix A: Methods motivating the new MDL
method}

This appendix uses the MDL principle to explain the statistical methods
that led to the MDL method defined by equation \eqref{eq:LFDR-MDL}.
This appendix also has results that lay the foundation for the operating
characteristics of the estimator given in Appendix B. A simple explanation
of basic MDL-theoretic ideas in terms of hypothesis testing is available
in the appendix of \citet{NMWL}. See \citet{RefWorks:374}, \citet{RefWorks:376},
\citet{RefWorks:375}, and \citet{support} for other introductions
to the MDL principle of model selection.

Since $\theta_{\text{{alt}}}$ is unknown, it will be replaced with
a parameter value chosen to minimize the codelength of the data according
to MDL theory, in which the length of a codeword is the number of
independently selected binary digits of equal probability that achieve
the joint probability of that codeword \citep{RefWorks:374}. The
availability of measurements pertaining to features other than the
inference target enables the construction of a universal codelength
function and a close approximation that is computationally more convenient.
The idea that statistical inference minimizes universal codelength
functions is called the \emph{MDL principle} and is often formalized
in terms of minimax problems.

\subsection*{Minimum description length concepts\label{sub:Universality}}

The theory of this section is presented in terms of a parametric family
that is free from nuisance parameters. In many cases, such a family
can be derived using one of the data reduction methods of Section
\ref{sec:Data-reduction}. 

Under the MDL framework, each scheme $\dagger$ for coding the data
under the alternative hypothesis corresponds to a codelength function
$L^{\dagger}$ on $\mathcal{X}$ and thus to a \emph{compressing probability
density function} $g^{\dagger}$ selected from the parametric family
$\left\{ g_{\theta}:\theta\in\Theta\right\} $ before observing $T\left(x\right)$,
the realized value of the statistic, with the goal of minimizing the
codelength $L^{\dagger}\left(T\left(x\right)\right)=-\log g^{\dagger}\left(T\left(x\right)\right)$.
Since $\theta_{0}$ is known, the probability density function of
the statistic under the null hypothesis is known to be $g_{\theta_{0}}$,
which compresses the data with respect to the null model. Accordingly,
the codelength function $L^{0}$ relative to the null hypothesis is
that specified by $L^{0}\left(T\left(x\right)\right)=-\log g_{\theta_{0}}\left(T\left(x\right)\right)$.
Since the base of the logarithm is arbitrary, the inverse logarithm
is denoted by $\log^{-1}\bullet$ rather than by $\exp\bullet$ or
by $2^{\bullet}$. 

Suppose, as in Example \ref{exa:absT-stat}, that there is a vector
$x_{i}$ of measurements for each of the $N$ features and that the
data are reduced to the statistics $T\left(x_{1}\right),\dots,T\left(x_{N}\right)$.
With $L_{i}^{\dagger}\left(T\left(x_{i}\right)\right)$ as the codelength
of $T\left(x_{i}\right)$ relative to the alternative hypothesis,
$\Delta_{i}^{\dagger}\left(T\left(x_{i}\right)\right)=L_{i}^{\dagger}\left(T\left(x_{i}\right)\right)-L^{0}\left(T\left(x_{i}\right)\right)$
is the \emph{information in $T\left(x_{i}\right)$ for discrimination}
favoring the null hypothesis over the alternative hypothesis; cf.
\citet{NMWL,support}. A difference in null and alternative codelengths
has been called a {}``universal test statistic'' \citep{RefWorks:342};
however, that term can cause confusion with $T\left(X_{i}\right)$.

\begin{example}
If the restriction to a parametric family were relaxed,
\begin{equation}
-\log\frac{\hat{g}_{\text{alt}}\left(T\left(x_{i}\right)\right)}{g_{\theta_{0}}\left(T\left(x_{i}\right)\right)}=-\log\frac{1-\widehat{\lfdr}\left(x_{i}\right)}{\widehat{\lfdr}\left(x_{i}\right)}+\log\frac{1-\hat{\pi}_{0}}{\hat{\pi}_{0}}\label{eq:eBayes}
\end{equation}
would be the information for discrimination according to the empirical
Bayes methodology of Section \ref{sec:Introduction}.\end{example}

The \emph{regret} \citep{RefWorks:375} of the codelength function
$L_{i}^{\dagger}$ is 
\[
\reg\left(g_{i}^{\dagger},x_{i}\right)=L_{i}^{\dagger}\left(T\left(x_{i}\right)\right)-\inf_{\theta\in\Theta}\left(-\log g_{\theta}\left(T\left(x_{i}\right)\right)\right)=-\log\frac{g_{i}^{\dagger}\left(T\left(x_{i}\right)\right)}{g_{\hat{\theta}}\left(T\left(x_{i}\right)\right)},
\]
where $L_{i}^{\dagger}$ is given by $L_{i}^{\dagger}\left(T\left(x_{i}\right)\right)=-\log g_{i}^{\dagger}\left(T\left(x_{i}\right)\right)$
and where $\hat{\theta}=\arg\sup_{\theta\in\Theta}g_{\theta}\left(T\left(x\right)\right)$.
Likewise, the regret of the codelength function relative to the null
hypothesis is $\reg\left(g_{\theta_{0}},x_{i}\right)=-\log\left(g_{\theta_{0}}\left(T\left(x_{i}\right)\right)/g_{\hat{\theta}}\left(T\left(x_{i}\right)\right)\right)$.

While the sign of $\Delta_{i}^{\dagger}\left(T\left(x_{i}\right)\right)$
indicates which hypothesis is favored \citep{RefWorks:342}, it can
also be compared to a threshold $J$ of the minimum amount of information
considered sufficient for selecting one hypothesis over the other.
In that case, the probability of observing misleading information
for discrimination has an upper bound for any distributions \textbf{$g_{\theta_{0}}$}
and $g_{i}^{\dagger}$. Specifically, for any $J>0,$
\begin{equation}
P_{\theta_{0},\lambda}\left(\Delta_{i}^{\dagger}\left(T\left(X_{i}\right)\right)\le-J\right)=P_{\theta_{0},\lambda}\left(g_{i}^{\dagger}\left(T\left(X_{i}\right)\right)/g_{\theta_{0}}\left(T\left(X_{i}\right)\right)\ge\log^{-1}J\right)\le1/\log^{-1}J.\label{eq:universal-bound}
\end{equation}

Applications to the probability of observing misleading evidence appear
in \citet{RefWorks:123}. A derivation from the Markov inequality
appears in \citet{hierarchy}. Since the derivation assumes that \textbf{$g_{\theta_{0}}$}
and $g_{i}^{\dagger}$ are genuine probability density functions,
formula \eqref{eq:universal-bound} does not necessarily hold for
pseudo-likelihoods such as profile likelihoods and likelihoods integrated
with respect to an improper prior; however, it does hold for all marginal
and conditional likelihoods \citep{RefWorks:123}.

The following two schemes ($\dagger$ and $\ddagger$) for coding
the reduced data give essentially identical regrets for a sufficiently
large value of $N.$

\subsubsection*{\label{sub:Exact-codelength}Exact codelength}

While the codelength function $L_{i}^{\dagger}$ for the $i$th feature
cannot depend on $x_{i}$, it may depend on $x_{j}$ for all $j\ne i$
as follows. For all $i=1,\dots,N$, define $L_{i}^{\dagger}$ such
that the corresponding probability density function $g_{i}^{\dagger}$
is equal to $g_{\theta_{i}^{\dagger}}$ for the value $\theta_{i}^{\dagger}$
such that

\begin{equation}
\theta_{i}^{\dagger}=\arg\inf_{\theta\in\Theta}\sum_{j\ne i}\min\left(\reg\left(g_{\theta},x_{j}\right),\reg\left(g_{\theta_{0}},x_{j}\right)\right).\label{eq:exact-parameter}
\end{equation}

In words, the code for a given feature uses the distribution in the
parametric family that minimizes the regret summed over all other
features. 

Proportional to $N^{2},$ the computation time can prohibit the use
of the universal compression method for large $N$. For example, $N$
can be in the tens of thousands for gene expression microarrays or
in the hundreds of thousands for genome-wide association studies.
The next coding scheme overcomes this problem because its computation
time is proportional to $N$.

\subsubsection*{Approximate codelength}

The $\dagger$ coding scheme is efficiently approximated by a slightly
illegal scheme denoted by $\ddagger$. It determines the codelength
for statistic $T\left(x_{i}\right)$ under the alternative hypothesis
by using a common probability density function $g^{\ddagger}$ that
is in the parametric family, i.e., $g^{\ddagger}=g_{\theta^{\ddagger}}$
for some $\theta^{\ddagger}\in\Theta$. This is accomplished by minimizing
the regret over all features
\begin{equation}
\theta^{\ddagger}=\arg\inf_{\theta\in\Theta}\sum_{j=1}^{N}\min\left(\reg\left(g_{\theta},x_{j}\right),\reg\left(g_{\theta_{0}},x_{j}\right)\right).\label{eq:approximate-parameter}
\end{equation}
This coding scheme is technically illegal in the sense that $g^{\ddagger}$,
as a function of the observed data for each feature, depends on hindsight.
However, under general conditions, $\theta^{\ddagger}$ approximates
$\theta_{i}^{\dagger}$ for all $i=1,\dots,N$ given sufficiently
large $N$ because the selection of the distribution depends on all
features without giving undue weight to any single feature. The approximation
is supported by the fact that both $\theta^{\dagger}$ and $\theta^{\ddagger}$
are maximum likelihood estimates of $\theta$ under the alternative
hypothesis:
\begin{thm}
\label{thm:approximate-code-MLE}Assume that for some $\theta_{0}\in\Theta$
and $\theta_{\text{alt}}\in\Theta$ such that $\theta_{0}\ne\theta_{\text{alt}}$
and that for all $j\in\left\{ 1,\dots,N\right\} $, each statistic
$T\left(X_{j}\right)$ has probability density $g_{\theta_{j}}$ with
$\theta_{j}\in\left\{ \theta_{0},\theta_{\text{alt}}\right\} $ and
is independent of every $T\left(X_{k}\right)$ with $k\in\left\{ 1,\dots,N\right\} \backslash\left\{ j\right\} $.
It follows that $\theta^{\ddagger}$, if unique, is the maximum likelihood
estimate of $\theta_{\text{alt}}$. \end{thm}
\begin{proof}
Using equation \eqref{eq:approximate-parameter},
\begin{eqnarray*}
\theta^{\ddagger} & = & \arg\inf_{\theta}\sum_{j=1}^{N}\min\left(-\log g_{\theta}\left(T\left(x_{j}\right)\right),\:-\log g_{\theta_{0}}\left(T\left(x_{j}\right)\right)\right)\\
 & = & \arg\sup_{\theta\in\Theta}\sum_{j=1}^{N}\max\left(\log g_{\theta}\left(T\left(x_{j}\right)\right),\:\log g_{\theta_{0}}\left(T\left(x_{j}\right)\right)\right)\\
 & = & \arg\sup_{\theta\in\Theta}\sup_{\boldsymbol{\theta}\in\left\{ \theta_{0},\theta_{\text{alt}}\right\} ^{N}}\,\sum_{j=1}^{N}\log g_{\theta_{j}}\left(T\left(x_{j}\right)\right)\\
 & = & \arg\sup_{\theta\in\Theta}\sup_{\boldsymbol{\theta}\in\left\{ \theta_{0},\theta_{\text{alt}}\right\} ^{N}}\,\prod_{j=1}^{N}g_{\theta_{j}}\left(T\left(x_{j}\right)\right),
\end{eqnarray*}
where $\boldsymbol{\theta}=\left\langle \theta_{1},\dots,\theta_{N}\right\rangle $
and $\left\{ \theta_{0},\theta_{\text{alt}}\right\} ^{N}$ is the
$N$-factor Cartesian product $\left\{ \theta_{0},\theta_{\text{alt}}\right\} \times\cdots\times\left\{ \theta_{0},\theta_{\text{alt}}\right\} $.\end{proof}
\begin{cor}
\label{cor:exact-code-MLE}Under the assumptions of Theorem \ref{thm:approximate-code-MLE},
$i\in\left\{ 1,\dots,N\right\} $, if $\theta_{i}^{\dagger}$ is unique,
then it is the maximum likelihood estimate of $\theta_{\text{alt}}$
on the basis of the outcomes $X_{j}=x_{j}$ for all $j\in\left\{ 1,\dots,N\right\} \backslash\left\{ i\right\} $. \end{cor}
\begin{proof}
The claim reduces to that of Theorem \ref{thm:approximate-code-MLE}
because the data are equivalent except for the presence or absence
of the outcome $T\left(X_{i}\right)=T\left(x_{i}\right)$ and because
$\theta_{i}^{\dagger}$ and $\theta^{\ddagger}$ are equivalent, except
for the presence or absence of the term involving that outcome. Thus,
for all $i\in\left\{ 1,\dots,N\right\} $, 
\[
\theta_{i}^{\dagger}=\arg\sup_{\theta\in\Theta}\sup_{\boldsymbol{\theta}\in\left\{ \theta_{0},\theta_{\text{alt}}\right\} ^{N}}\,\prod_{j\ne i}g_{\theta_{j}}\left(T\left(x_{j}\right)\right).
\]

\end{proof}
Theorem \ref{thm:approximation} of the next subsection specifies
sufficient conditions for the convergence of $\theta^{\ddagger}-\theta_{i}^{\dagger}$
to 0 as $N$ increases. 

The coding method of the section entitled {}``Exact codelength''
is \emph{universal} in the sense that it asymptotically compresses
the data as much as the noiseless coding theorem allows for any distribution
in the parametric family (cf. \citet[§3.7]{RefWorks:374} and \citet[§6.5]{RefWorks:375}).
Sufficient conditions for universality are stated in the following
lemma, in which \emph{strong consistency} means almost sure convergence
to a parameter value as $n\rightarrow\infty$ if each $T\left(X_{i}\right)$
is stationary and, at fixed $n$, of a density function in $\left\{ g_{\theta}:\theta\in\Theta\right\} $.
(The dependence of $g_{\theta}$ on $n$ is suppressed.) Such convergence
will be denoted by $\overset{n}{\rightarrow}$.
\begin{lem}[Consistency]
\label{lem:consistency} Suppose that for some $\theta_{0}\in\Theta$
and $\talt\in\Theta$ such that $\theta_{0}\ne\talt$ and that for
all $j\in\left\{ 1,\dots,N\right\} $, each statistic $T\left(X_{j}\right)$
has probability density $g_{\theta_{j}}$ with $\theta_{j}\in\left\{ \theta_{0},\talt\right\} $
such that $\theta_{j}=\talt$ for at least two values of $j$ in $\left\{ 1,\dots,N\right\} $.
Suppose further that $g_{\bullet}\left(T\left(X_{j}\right)\right)$
is almost surely continuous on $\Theta$ for all $j\in\left\{ 1,\dots,N\right\} $.
If, for some $i\in\left\{ 1,\dots,N\right\} $, $\theta_{i}^{\dagger}$
is unique and $\hat{\theta}_{j}=\arg\sup_{\theta\in\Theta}g_{\theta}\left(T\left(X_{j}\right)\right)$
is a strongly consistent estimate of $\theta_{j}$ for all $j\in\left\{ 1,\dots,N\right\} \backslash\left\{ i\right\} $,
then $\theta_{i}^{\dagger}$ is a strongly consistent estimate of
$\talt$.\end{lem}
\begin{proof}
Let $\mathfrak{J}=\left\{ j:\theta_{j}=\theta_{\text{alt}},j\in\left\{ 1,\dots,N\right\} \backslash\left\{ i\right\} \right\} $,
which by assumption is nonempty. By the consistency condition, $\hat{\theta}_{j}\overset{n}{\rightarrow}\theta_{\text{alt}}$
for all $j\in\mathfrak{J}$ and $\hat{\theta}_{j}\overset{n}{\rightarrow}\theta_{0}$
for all $j\in\left\{ 1,\dots,N\right\} \backslash\mathfrak{J}$. Thus,
with probability 1,
\begin{eqnarray*}
\prod_{j\ne i}g_{\theta_{j}}\left(T\left(X_{j}\right)\right) & = & \prod_{j\in\mathfrak{J}}g_{\theta_{\text{alt}}}\left(T\left(X_{j}\right)\right)\prod_{j\notin\mathfrak{J}\cup\left\{ i\right\} }g_{\theta_{0}}\left(T\left(X_{j}\right)\right)\\
 & = & \prod_{j\ne i}g_{\hat{\theta}_{j}}\left(T\left(X_{j}\right)\right)\\
 & = & \prod_{j\ne i}\max\left(g_{\theta_{\text{alt}}}\left(T\left(X_{j}\right)\right),g_{\theta_{0}}\left(T\left(X_{j}\right)\right)\right)\\
 & = & \sup_{\theta\in\Theta}\prod_{j\ne i}\max\left(g_{\theta}\left(T\left(X_{j}\right)\right),g_{\theta_{0}}\left(T\left(X_{j}\right)\right)\right)
\end{eqnarray*}
in the limit as $n\rightarrow\infty$, with the equalities holding
by the almost-sure continuity of $g_{\bullet}\left(T\left(X_{j}\right)\right)$
as a function on $\Theta$ \citep[§1.7]{Serfling:1254212}. Because
by equation \eqref{eq:exact-parameter}, 
\[
\theta_{i}^{\dagger}=\arg\sup_{\theta\in\Theta}\sum_{j\ne i}\max\left(g_{\theta}\left(T\left(X_{j}\right)\right),g_{\theta_{0}}\left(T\left(X_{j}\right)\right)\right),
\]
it follows that $\theta_{i}^{\dagger}\overset{n}{\rightarrow}\theta_{i}$. 
\end{proof}
Heuristically, the key observation of the proof is that whether $\theta$
is constrained to have one of the two values has no asymptotic effect
on the estimates of $\theta_{j}$. The universality of the codelength
function is a consequence.\begin{thm}[Universality]
\label{thm:universality} Under the conditions of Lemma \ref{lem:consistency},
\[
\lim_{n\rightarrow\infty}E_{\talt}\left(\frac{L_{i}^{\dagger}\left(T\left(X_{i}\right)\right)}{n}\right)=\lim_{n\rightarrow\infty}E_{\talt}\left(\frac{-\log g_{\theta_{\text{alt}}}\left(T\left(X_{i}\right)\right)}{n}\right)
\]
for all $i\in\left\{ 1,\dots,N\right\} $ such that $\theta_{i}=\talt$,
where $E_{\talt}$ signifies the expectation value with respect to
$g_{\talt}$, i.e., $E_{\talt}\left(\bullet\right)=\int\bullet dP_{\talt}$.\end{thm}
\begin{proof}
$P_{\theta_{\text{alt}}}\left(\lim_{n\rightarrow\infty}\theta_{i}^{\dagger}\in\left\{ \theta_{0},\theta_{\text{alt}}\right\} \right)=1$
for all $i\in\left\{ 1,\dots,N\right\} $ because $\theta_{i}^{\dagger}\overset{n}{\rightarrow}\theta_{i}$
by the lemma and $\theta_{i}\in\left\{ \theta_{0},\theta_{\text{alt}}\right\} $
by assumption. Hence, $\theta_{i}^{\dagger}\overset{n}{\rightarrow}\talt$
for all $i\in\left\{ 1,\dots,N\right\} $ such that $\theta_{i}=\theta_{\text{alt}}$.
Thus, for those values of $i$, 
\[
\lim_{n\rightarrow\infty}E_{\talt}\left(\frac{-\log\left(g_{\talt}\left(T\left(X_{i}\right)\right)/g_{\theta_{i}^{\dagger}}\left(T\left(X_{i}\right)\right)\right)}{n}\right)=0
\]
because $g_{\talt}\left(T\left(X_{i}\right)\right)/g_{\theta_{i}^{\dagger}}\left(T\left(X_{i}\right)\right)\overset{n}{\rightarrow}1$
by the almost-sure continuity of $g_{\bullet}\left(T\left(X_{i}\right)\right)$
as a function on $\Theta$ \citep[§1.7]{Serfling:1254212}. 
\end{proof}
The $N\rightarrow\infty$ universally of a related mixture code will
be established in Appendix B.

\subsection*{Asymptotic characteristics of\textmd{\normalsize{} $\theta^{\ddagger}$
and $\theta_{i}^{\dagger}$}}

Assume $X_{1},X_{2},\dots$ are independent and each of identical
distribution $P_{\star}$. For example, $P_{\star}$ could be a \emph{K}-component
mixture distribution $P_{\star}=\sum_{k=1}^{K}\pi_{k}P_{\star k},$
where $\pi_{k}$ is the probability that some $X_{j}$ has distribution
$P_{\star k}$, which is not necessarily in $\left\{ P_{\theta,\lambda}:\theta\in\Theta,\lambda\in\Lambda\right\} $.
Let $E_{\star}\left(\bullet\right)$ and $\overset{N}{\rightarrow}$
denote the expectation value and almost-sure convergence as $N\rightarrow\infty$
with respect to $P_{\star}$.
\begin{thm}
\label{thm:approximation}\emph{Suppose that}\textup{\emph{,}}\emph{
for all} $i\in\left\{ 1,\dots,N\right\} $\textup{, $E_{\star}\left(\log g_{\theta}\left(T\left(X_{j}\right)\right)\right)<\infty$
for all $\theta\in\Theta$ and that $\theta^{\ddagger}$ and $\theta_{i}^{\dagger}$
are unique with }$P_{\star}$\textup{-probability 1.} \emph{Then}
$\theta^{\ddagger}-\theta_{i}^{\dagger}\overset{N}{\rightarrow}0$
for all $i\in\left\{ 1,\dots,N\right\} $.\end{thm}
\begin{proof}
For any $\theta\in\Theta$, let $\hat{\theta}_{j}\left(\theta\right)=\arg\max_{\widetilde{\theta}\in\left\{ \theta_{0},\theta\right\} }g_{\widetilde{\theta}}\left(T\left(x_{j}\right)\right)$.
Because $\log g_{\hat{\theta}_{j}\left(\theta\right)}\left(T\left(X_{j}\right)\right)$
is IID for all $j\in\left\{ 1,\dots,N\right\} $, the strong law of
large numbers implies that, for all $\mathcal{J}_{N}\in\left\{ \left\{ 1,\dots,N\right\} ,\left\{ 1,\dots,N\right\} \backslash\left\{ 1\right\} ,\dots,\left\{ 1,\dots,N\right\} \backslash\left\{ N\right\} \right\} $,
\[
\frac{1}{\left|\mathcal{J}_{N}\right|}\sum_{j\in\mathcal{J}_{N}}\log g_{\hat{\theta}_{j}\left(\theta\right)}\left(T\left(X_{j}\right)\right)\overset{N}{\rightarrow}E_{\star}\left(\log g_{\hat{\theta}_{j}\left(\theta\right)}\left(T\left(X_{j}\right)\right)\right)
\]
\[
=P_{\star}\left(\hat{\theta}_{j}\left(\theta\right)=\theta_{0}\right)E_{\star}\left(\log g_{\hat{\theta}_{j}\left(\theta\right)}\left(T\left(X_{j}\right)\right)\vert\hat{\theta}_{j}\left(\theta\right)=\theta_{0}\right)
\]
\[
+P_{\star}\left(\hat{\theta}_{j}\left(\theta\right)=\theta\right)E_{\star}\left(\log g_{\hat{\theta}_{j}\left(\theta\right)}\left(T\left(X_{j}\right)\right)\vert\hat{\theta}_{j}\left(\theta\right)=\theta\right),
\]
the finiteness of which follows from that of $E_{\star}\left(\log g_{\theta}\left(T\left(X_{j}\right)\right)\right)$.
 As the result holds for arbitrary $\theta\in\Theta$, 
\[
\arg\sup_{\theta\in\Theta}\frac{1}{\left|\mathcal{J}_{N}\right|}\sum_{j\in\mathcal{J}_{N}}\log g_{\hat{\theta}_{j}\left(\theta\right)}\left(T\left(X_{j}\right)\right)\overset{N}{\rightarrow}
\]
\[
\arg\sup_{\theta\in\Theta}E_{\star}\left(\log g_{\hat{\theta}_{i}\left(\theta\right)}\left(T\left(X_{j}\right)\right)\right)
\]
irrespective of whether the sum on the left-hand-side is over $\left\{ 1,\dots,N\right\} $
or over $\left\{ 1,\dots,N\right\} \backslash\left\{ i\right\} $
for some $i\in\left\{ 1,\dots,N\right\} $. (The uniqueness of the
maximizing value of $\theta$ on the left-hand-side is guaranteed
by the postulated uniqueness of $\theta^{\ddagger}$ and $\theta_{i}^{\dagger}$.)
Therefore,  the difference in the maximum likelihood estimate of
$\theta$ under the alternative hypothesis using $X_{1},\dots,X_{N}$
and that using $X_{1},\dots,X_{i-1},X_{i+1},\dots,X_{N}$ converges
almost surely to 0; however, such maximum likelihood estimates are
$\theta^{\ddagger}$ and $\theta_{i}^{\dagger}$, respectively, according
to Theorem \ref{thm:approximate-code-MLE} and Corollary \ref{cor:exact-code-MLE}. 
\end{proof}

\section*{Appendix B: Asymptotic characteristics of MDL and L0O}

This section extends the fixed-component results of Appendix A to
the two-component mixture density of equation \eqref{eq:mixture-density}
with the constraint that $g_{\text{alt}}=g_{\theta_{\text{alt}}}$
for some $\theta_{\text{alt}}\in\Theta$. In this setting, the universal
density $g_{i}^{\dagger}$ and its approximation $g^{\ddagger}$ are
replaced with $g_{i}^{\mdl}=g_{\theta_{i}^{\mdl}}$ and its approximation
$g^{\loo}=g_{\theta^{\loo}}$, where $\left\langle \theta_{i}^{\mdl},\pi_{0i}^{\mdl}\right\rangle $
are given by equation \eqref{eq:maxlik-MDL}. (\citet{GWAselect}
compared the performance of $g^{\ddagger}$ and $g^{\loo}$ by simulation.)

Assuming the statistics are independent, $\left\langle \theta_{i}^{\mdl},\pi_{0i}^{\mdl}\right\rangle $
and $\left\langle \theta^{\loo},\pi_{0}^{\loo}\right\rangle $ are
clearly maximum likelihood estimates of $\left\langle \theta_{\text{alt}},\pi_{0}\right\rangle $.
Consequently, the steps used to prove Theorem \ref{thm:approximation}
also demonstrate that $\theta^{\loo}-\theta_{i}^{\mdl}\overset{N}{\rightarrow}0$
and $\pi_{0}^{\loo}-\pi_{0i}^{\mdl}\overset{N}{\rightarrow}0$ for
all $i\in\left\{ 1,\dots,N\right\} $ under the independence condition.
The mixture codes form LFDR estimates via substituting either $\theta^{\mdl}$
and $\pi_{0}^{\mdl}$ or $\theta^{\loo}$ and $\pi_{0}^{\loo}$ into
equations \eqref{eq:LFDR} and \eqref{eq:mixture-density}. 

Whereas regularity conditions entailing the strong consistency of
maximum likelihood estimates for finite-mixture models \citep{RednerWalker1984b}
would apply as $N\rightarrow\infty$, seemingly more pertinent to
universality is consistency in the sense of $\overset{n}{\rightarrow}$,
which is almost-sure convergence as $n\rightarrow\infty$ under the
stationarity of every $T\left(X_{i}\right)$. However, such $\overset{n}{\rightarrow}$
consistency does not hold if $N$ is finite and if $\pi_{0}>0$, for
in that case, there is fixed, nonzero probability $\pi_{0}^{N}$ that
all $N$ statistics have probability density function $g_{\theta_{0}}$
rather than $g_{\theta_{\text{alt}}}$. Therefore, $\overset{N}{\rightarrow}$
consistency will be used instead.\begin{thm}
\label{thm:universality-asymptotic-1}If the maximum likelihood estimate
$\theta^{\loo}$ almost surely converges to $\theta_{\text{alt}}$
as $N\rightarrow\infty$ and if\textup{ }$g_{\bullet}\left(T\left(X_{i}\right)\right)$
is almost surely continuous on $\Theta$ for all $i\in\left\{ 1,2,\dots\right\} $,
then 
\[
\lim_{N\rightarrow\infty}E_{\talt}\left(L_{i}^{\ast}\left(T\left(X_{i}\right)\right)/n\right)=E_{\talt}\left(-\log g_{\talt}\left(T\left(X_{i}\right)\right)/n\right)
\]
for all $i\in\left\{ 1,2,\dots\right\} $ such that $\theta_{i}=\talt$,
where $L_{i}^{\ast}\left(T\left(X_{i}\right)\right)=-\log g_{\theta_{i}^{\mdl}}\left(T\left(X_{j}\right)\right)$
and $E_{\talt}$ signifies the expectation value with respect to $g_{\talt}$,
i.e., $E_{\talt}\left(\bullet\right)=\int\bullet dP_{\talt}$.\end{thm}
\begin{proof}
Since $\theta_{i}^{\mdl}$ is the maximum likelihood estimate for
the $N-1$ statistics other than $T\left(X_{i}\right)$, $\theta_{i}^{\mdl}\overset{N}{\rightarrow}\theta_{\text{alt}}.$
Thus, the claim follows from reasoning analogous to that used to prove
Theorem \ref{thm:universality}.\end{proof}
\begin{cor}[Asymptotic universality]
Given the conditions of Theorem \ref{thm:universality-asymptotic-1},
\[
\lim_{n\rightarrow\infty}\lim_{N\rightarrow\infty}E_{\talt}\left(\frac{L_{i}^{\ast}\left(T\left(X_{i}\right)\right)}{n}\right)
\]
\[
=\lim_{n\rightarrow\infty}E_{\talt}\left(\frac{-\log g_{\talt}\left(T\left(X_{j}\right)\right)}{n}\right)
\]
for all $i\in\left\{ 1,2,\dots\right\} $ such that $\theta_{i}=\talt$.
\end{cor}
The proof is trivial. The corollary means that 
\[
\left(L_{i}^{\ast}\left(T\left(x_{i}\right)\right)-\log\left(1-\pi_{0i}^{\mdl}\right)\right)-\left(L^{0}\left(T\left(x_{i}\right)\right)-\log\pi_{0i}^{\mdl}\right)
\]
may be regarded as approaching the information for discrimination
under the mixture model as $N\rightarrow\infty$. Since $\theta^{\loo}-\theta_{i}^{\mdl}\overset{N}{\rightarrow}0$
and $\pi_{0}^{\loo}-\pi_{0i}^{\mdl}\overset{N}{\rightarrow}0$, that
information is approximated by substituting the maximum likelihood
estimates $\theta^{\loo}$ and $\pi_{0}^{\loo}$ for $\theta_{i}^{\mdl}$
and $\pi_{0i}^{\mdl}$, respectively.

\bibliographystyle{elsarticle-harv}
\bibliography{refman,Marta}

\end{document}